\documentclass[11pt]{article}
\usepackage{latexsym}
\usepackage{theorem}
\usepackage{graphicx}
\usepackage{amsmath,color}
\usepackage{amsfonts}
\usepackage{natbib}
\usepackage{soul}

\usepackage{tikz}
\usepackage{caption}
\usepackage{subcaption}

\theorembodyfont{\rmfamily}
\newtheorem{theorem}{Theorem}
\newtheorem{conjecture}[theorem]{Conjecture}
\newtheorem{lemma}[theorem]{Lemma}
\newtheorem{proposition}[theorem]{Proposition}

\newtheorem{definition}[theorem]{Definition}

\theoremstyle{break}
\newtheorem{remark}[theorem]{Remark}

\usepackage{mathtools}

\newenvironment{proof}{\paragraph{Proof.}}{\hfill$\square$}

\title{Optimal Patrolling Strategies for \\Trees and Complete Networks}

\date{}
\author{Thuy Bui\thanks{Rutgers Business School, 1 Washington Park, Newark, NJ 07102, USA, tb680@business.rutgers.edu} \and Thomas Lidbetter\thanks{Department of Engineering Systems and Environment, University of Virginia, VA 22903, USA, tlidbetter@virginia.edu} \thanks{Rutgers Business School, 1 Washington Park, Newark, NJ 07102, USA, tlidbetter@business.rutgers.edu}}

\providecommand{\keywords}[1]{\textbf{\textbf{Keywords:}} #1}

\begin{document}
	
\maketitle

\begin{abstract}
\noindent We present solutions to a continuous patrolling game played on network. 
In this zero-sum game, an Attacker chooses a time and place to attack a network for a fixed amount of time. 
A Patroller patrols the network with the aim of intercepting the attack with maximum probability. 
Our main result is the proof of a recent conjecture on the optimal patrolling strategy for trees. 
The conjecture asserts that a particular patrolling strategy called the {\em $E$-patrolling strategy} is optimal for all tree networks. 
The conjecture was previously known to be true in a limited class of special cases. 
The $E$-patrolling strategy has the advantage of being straightforward to calculate and implement.
We prove the conjecture by presenting $\varepsilon$-optimal strategies for the Attacker which provide upper bounds for the value of the game that come arbitrarily close to the lower bound provided by the $E$-patrolling strategy.
We also solve the patrolling game in some cases for complete networks.
\end{abstract}


\newpage
\section{Introduction}
In the continuous patrolling game, introduced by \cite{alpern2016patrolling}, an Attacker picks a point on a network $Q$ and a time interval of fixed duration during which to carry out an attack. A Patroller moves on the network at unit speed and intercepts the attack (and wins the game) if she reaches the attacked point during the attack interval. \cite{alpern2022continuous} proposed a mixed strategy for the Patroller, called the  {\em $E$-patrolling strategy}, which was shown to be optimal for certain classes of tree networks. In Conjecture~1 of that paper, they suggested that the $E$-patrolling strategy was optimal for all trees. We refer to this conjecture as the {\em tree patrolling conjecture}. In this paper we settle the tree patrolling conjecture by proving that the $E$-patrolling strategy is optimal for all tree networks. We also solve the game in certain cases for complete networks (those for which every pair of nodes is connected by precisely one arc).

The key idea we use to prove the conjecture for trees is that as long as the Attacker randomizes over a large enough time period, there are mixed strategies that are arbitrarily close to being optimal that simply pick the time of the attack uniformly over that period. 
This means that we need only specify a distribution over the network $Q$.
We define a mixed strategy for the Attacker that is played over a large time interval $[0,T]$ and show that for any given $\varepsilon >0$, this strategy is $\varepsilon$-optimal for large enough $T$.

Most work in the area of patrolling games focuses on discrete models, such as \cite{ALP}, \cite{AMP},  \cite{Lin13}, \cite{ARMOR}, \cite{Yolmeh} and \cite{Zoroa12}. A disadvantage of discrete models is that in many real world examples of patrolling, an attack or infiltration can occur anywhere continuously along a border, boundary or network. Discrete models also assume that attacks occur at discrete times, but of course it is more realistic to model time as continuous. This was the motivation behind the continuous patrolling game introduced by \cite{alpern2016patrolling}. As well as the recent work of \cite{alpern2022continuous} on the game, \cite{garrec2019continuous} has also made some important contributions, including establishing that the game has a value and optimal (or $\varepsilon$-optimal) strategies. \cite{Lin19} studied a different continuous patrolling game on a perimeter.

The layout of the paper is as follows. In Section~\ref{sec:bground}, we recall the definition of the continuous patrolling game and give some background on previous work on the game. We also describe the tree patrolling conjecture precisely. In Section~\ref{sec:roots}, we work towards defining a decomposition of any tree $Q$ which we call its {\em subtree decomposition}. This decomposition consists of a set of subtrees of length at most $\alpha/2$ containing all the leaf nodes and another connected set we call the {\em core}. We also define the concept of the {\em density} of a subset of a network, which, for a given Attacker strategy, is defined as the ratio of the probability the attack takes place in that subset to the length of the subset. This definition is analogous to the concept of {\em search density}, which is well known in the field of {\em search games}. The concept originates from the work of \cite{gal79}, but has been used more recently in, for example, \cite{alpern2013search}, \cite{fokkink19} and \cite{hermans22}. The ideas of density and the subtree decomposition are crucial for us to define in Section~\ref{sec:main} the Attacker strategy that we proceed to show is $\varepsilon$-optimal. In Section~\ref{sec:complete} we solve the game on complete networks for some values of $\alpha$. In Section~\ref{sec:concl} we conclude.

The significance of our main result on trees lies in the fact that the $E$-patrolling strategy is intuitive and easy to implement. Roughly speaking, the Patroller repeatedly tours the network, but performs extra tours of subtrees of the network that are close to the leaf nodes. 

\section{Background and Definitions} \label{sec:bground}

In this section we make some definitions and give some more background to the continuous patrolling game. We finish the section by stating the tree patrolling conjecture precisely.

We start by defining a network $Q$ in a little more detail, though we refer the reader to \cite{alpern2022continuous} for a precise definition. A network $Q$ is given by a multigraph whose arcs can be viewed as open intervals. The length of an arc $a$ is denoted $\lambda(a)$, and $\lambda$ is extended to define a measure on $Q$. At each end of an arc is a node, and we refer to points of $Q$ that are not nodes as {\em regular}. We also define a metric $d$ on $Q$, where $d(x,y)$ is the length of the shortest path between two points $x,y \in Q$. 

In the continuous patrolling game on $Q$, the Attacker picks a point $x \in Q$ and a time $t \ge 0$ at which to start the attack.  
The attack lasts for time $\alpha$, where $\alpha>0$ is some parameter of the problem known to both players, and is no greater than the minimum tour time of $Q$.
The Patroller picks a patrol of the network, which is given by a unit speed path $S:[0,\infty) \rightarrow Q$. 
If the patrol intercepts the attack, then the Patroller wins the game. 
More precisely, the payoff of the game is equal to~1 if $x \in S([t, t+\alpha])$, otherwise the payoff is 0. 
The Patroller is the maximizer and the Attacker is the minimizer.

As mentioned in the Introduction, the continuous patrolling game was introduced in \cite{alpern2016patrolling}. 
\cite{garrec2019continuous} later proved that this zero-sum game has a value; moreover that the Patroller has optimal mixed strategies and the Attacker has $\varepsilon$-optimal mixed strategies (that is strategies that ensure the expected payoff is within $\varepsilon$ of the value of the game, for any $\varepsilon>0$). Garrec also found optimal strategies in the game in some special cases, as did \cite{alpern2016patrolling}.

\cite{alpern2022continuous} solved the game in some further special cases. 
Firstly, they gave a solution for arbitrary networks as long as $\alpha$ is shorter than the length of any arc of the network.  
Secondly, they gave a solution for tree networks when $\alpha$ is such that a particular condition called the {\em  Leaf Condition} is satisfied. 
They defined a patrolling strategy called the $E$-patrolling strategy, and showed that it is optimal for trees that satisfy the Leaf Condition. 
They conjectured that the $E$-patrolling strategy is optimal for all tree networks (the tree patrolling conjecture). They verified their conjecture for a class of star networks consisting of one long arc and an arbitrary number of short arcs of equal length. They also verified it for one particular example of a tree network that is not a star and does not satisfy the Leaf Condition. 

Generally speaking, the Leaf Condition is satisfied when $\alpha$ is particular small and, in the case of star networks, also when it is particularly large. This leaves a sizeable gap of values of $\alpha$ for which the optimality of the $E$-patrolling strategy was unproven.
In Section~\ref{sec:main}, we settle the tree patrolling conjecture. 

Of crucial importance to stating and proving the tree patrolling conjecture, we must first define the {\em extremity set} $E$ for a tree network $Q$.

Let $Q$ be a tree network of length $\mu$. For any set of points $Y$, we denote $Y^c$ for $Q-Y$ and $\overline{Y}$ for the topological closure of $Y$. 
If $x$ is a regular point of $Q$, then $Q-\{x\}$ has two components $Q_1(x)$ and $Q_2(x)$ such that $\lambda(Q_1(x))+\lambda(Q_2(x))=\mu$, and $\text{min}_{i=1,2}Q_i(x)\leq \mu/2$. If $x$ is a node of degree $n$ ($n \geq 3$), then $Q-\{x\}$ has $n$ components.  
\begin{definition}
Let $Q$ be a tree. The extremity set $E \equiv E(Q,\alpha)$ is defined as the set of all regular points $x \in Q$ such that $\text{min}_{i=1,2}\lambda(Q_i(x)) <\alpha/2$.
\end{definition}

Although it is convenient to define $E$ as an open set, we will largely work with its topological closure $\overline{E}$. In Figure~\ref{fig1} we depict the set $\overline{E}$ in red for various values of $\alpha$ on a specific tree network $Q$ of length $\mu=10$. Note that $\overline{E}=Q$ for $\alpha \ge 8$, and it is easy to see that in fact for any tree network $Q$, we have $\overline{E}=Q$ for all $\alpha>\mu$.

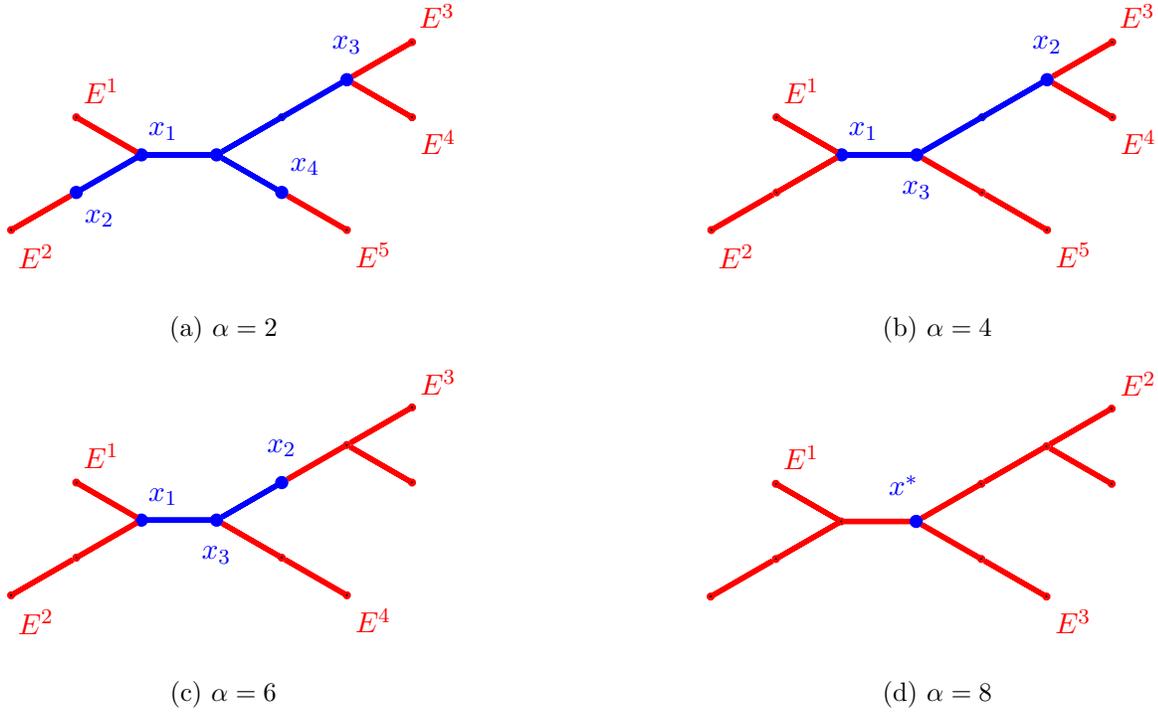
\begin{figure}[!ht]
    \centering
    \begin{subfigure}[b]{0.35\textwidth}
         \centering
        \begin{tikzpicture}[line width=1pt]
    \tikzstyle{every node}=[draw,circle,fill=black,minimum size=1pt,
                            inner sep=0pt]
    \draw (0,0) [color=blue] node (a)[label=60:${\ x_1\ }$, color=blue, minimum size=3pt] {}
        -- ++(0:1cm) [color= blue, line width=2pt] node (0) [label=-90:${}$, color=blue, minimum size=3pt] {}
        -- (a);
    \draw (0)
       -- ++(30:1cm) [color= blue, line width=2pt] node (12) {}
       -- ++(30:1cm) node (1) [label=90:${\ x_3\ }$, color=blue, minimum size=3pt]{}
       -- (0);
    \draw (0)
       -- ++(-30:1cm) [color= blue, line width=2pt] node (2) [label=60:${\ x_4\ }$, color=blue, minimum size=3pt]{}
       -- (0);
    \draw (2)
       -- ++(-30:1cm) [color= red, line width=2pt] node (22) [label=-30:${\ E^5\ }$]{}
       -- (2);
    \draw (1)
       -- ++(30:1cm) [color= red, line width=2pt] node (3) [label=60:${\ E^3\ }$]{}
       -- (1);
    \draw (1)
       -- ++(-30:1cm) [color= red, line width=2pt] node (4) [label=-30:${\ E^4\ }$]{}
       -- (1);
    \draw (a)
     -- ++(150:1cm)[color= red, line width=2pt] node (5) [label=60:${\ E^1\ }$]{}
     -- (a);
    \draw (a)
     -- ++(-150:1cm)[color= blue, line width=2pt] node (6)[label=-60:${\ x_2\ }$, color=blue, minimum size=3pt] {}
     -- (a);
    \draw(6)
     -- ++(-150:1cm)[color= red, line width=2pt] node (62) [label=-60:${\ E^2\ }$]{}
     -- (6);
\end{tikzpicture}
         \caption{$\alpha=2$}
         \label{fig:1}
     \end{subfigure}
    \hfill
    \begin{subfigure}[b]{0.5\textwidth}
         \centering
    \begin{tikzpicture}[line width=1pt]
    \tikzstyle{every node}=[draw,circle,fill=black,minimum size=1pt,
                            inner sep=0pt]
    \draw (0,0) [color=blue] node (a)[label=60:${\ x_1\ }$, color=blue, minimum size=3pt] {}
        -- ++(0:1cm) [color= blue, line width=2pt] node (0) [label=-90:${\ x_3\ }$, color=blue, minimum size=3pt] {}
        -- (a);
    \draw (0)
       -- ++(30:1cm) [color= blue, line width=2pt] node (12) {}
       -- ++(30:1cm) node (1) [label=90:${\ x_2\ }$, color=blue, minimum size=3pt]{}
       -- (0);
    \draw (0)
       -- ++(-30:1cm) [color= red, line width=2pt] node (2) {}
       -- (0);
    \draw (2)
       -- ++(-30:1cm) [color= red, line width=2pt] node (22) [label=-30:${\ E^5\ }$]{}
       -- (2);
    \draw (1)
       -- ++(30:1cm) [color= red, line width=2pt] node (3) [label=60:${\ E^3\ }$]{}
       -- (1);
    \draw (1)
       -- ++(-30:1cm) [color= red, line width=2pt] node (4) [label=-30:${\ E^4\ }$]{}
       -- (1);
    \draw (a)
     -- ++(150:1cm)[color= red, line width=2pt] node (5) [label=60:${\ E^1\ }$]{}
     -- (a);
    \draw (a)
     -- ++(-150:1cm) [color= red, line width=2pt] node (6) {}
     -- (a);
    \draw(6)
     -- ++(-150:1cm)[color= red, line width=2pt] node (62) [label=-60:${\ E^2\ }$]{}
     -- (6);
\end{tikzpicture}
         \caption{$\alpha=4$}
         \label{fig:2}
     \end{subfigure}
     
    \begin{subfigure}[b]{0.35\textwidth}
         \centering
        \begin{tikzpicture}[line width=1pt]
    \tikzstyle{every node}=[draw,circle,fill=black,minimum size=1pt,
                            inner sep=0pt]
    \draw (0,0) [color= blue] node (a)[label=60:${\ x_1\ }$, color=blue, minimum size=3pt] {}
        -- ++(0:1cm) [color= blue, line width=2pt] node (0) [label=-90:${\ x_3\ }$, color=blue, minimum size=3pt] {}
        -- (a);
    \draw (0)
       -- ++(30:1cm) [color= blue, line width=2pt]node (1) [label=90:${\ x_2\ }$, color=blue, minimum size=3pt]{}
       -- (0);
    \draw (1)
       -- ++(30:1cm) [color= red, line width=2pt] node (12) {}
       -- (1);
    \draw (0)
       -- ++(-30:1cm) [color= red, line width=2pt] node (2) {}
       -- (0);
    \draw (2)
       -- ++(-30:1cm) [color= red, line width=2pt] node (22) [label=-30:${\ E^4\ }$]{}
       -- (2);
    \draw (12)
       -- ++(30:1cm) [color= red, line width=2pt] node (3) [label=30:${\ E^3\ }$]{}
       -- (12);
    \draw (12)
       -- ++(-30:1cm) [color= red, line width=2pt] node (4){}
       -- (12);
    \draw (a)
     -- ++(150:1cm)[color= red, line width=2pt] node (5) [label=60:${\ E^1\ }$]{}
     -- (a);
    \draw (a)
     -- ++(-150:1cm) [color= red, line width=2pt] node (6) {}
     -- (a);
    \draw(6)
     -- ++(-150:1cm)[color= red, line width=2pt] node (62) [label=-60:${\ E^2\ }$]{}
     -- (6);
\end{tikzpicture}
         \caption{$\alpha=6$}
         \label{fig:3}
     \end{subfigure}
    \hfill
    \begin{subfigure}[b]{0.5\textwidth}
         \centering
    \begin{tikzpicture}[line width=1pt]
    \tikzstyle{every node}=[draw,circle,fill=black,minimum size=1pt,
                            inner sep=0pt]
    \draw (0,0) [color= blue] node (a) [label=10:${\ \ \ \ \ \ x^*\ }$]{}
        -- ++(0:1cm) [color= red, line width=2pt] node (0) [color=blue, minimum size=3pt] {}
        -- (a);
    \draw (0)
       -- ++(30:1cm) [color= red, line width=2pt] node (1) {}
       -- (0);
    \draw (1)
       -- ++(30:1cm) [color= red, line width=2pt] node (12) {}
       -- (1);
    \draw (0)
       -- ++(-30:1cm) [color= red, line width=2pt] node (2) {}
       -- (0);
    \draw (2)
       -- ++(-30:1cm) [color= red, line width=2pt] node (22) [label=-30:${\ E^3\ }$]{}
       -- (2);
    \draw (12)
       -- ++(30:1cm) [color= red, line width=2pt] node (3) [label=30:${\ E^2\ }$]{}
       -- (12);
    \draw (12)
       -- ++(-30:1cm) [color= red, line width=2pt] node (4){}
       -- (12);
    \draw (a)
     -- ++(150:1cm)[color= red, line width=2pt] node (5) [label=60:${\ E^1\ }$]{}
     -- (a);
    \draw (a)
     -- ++(-150:1cm) [color= red, line width=2pt] node (6) {}
     -- (a);
    \draw(6)
     -- ++(-150:1cm)[color= red, line width=2pt] node (62) {}
     -- (6);
\end{tikzpicture}
         \caption{$\alpha=8$}
         \label{fig:4}
     \end{subfigure}
    \caption{The components of $\overline{E}$ are shown in red and the core $E^0$ is shown in blue for $\alpha= 2, 4, 6, 8$. The local roots $x^*, x_i$ ($i \geq 1$) are labeled as blue points. }
    \label{fig1}
\end{figure}

We make a number of observations about $\overline{E}$, which we state without proof.

\begin{proposition} \label{prop:alpha} Let $Q$ be a tree. Then

(i)  $\overline{E}(Q,\alpha_1) \subseteq \overline{E}(Q, \alpha_2)$ for any $\alpha_1 \leq \alpha_2$;

(ii) there exists an unique $\alpha^*$ such that $\overline{E}(Q, \alpha^*)=Q$ and $\overline{E}(Q,\alpha)\neq Q$ for any $\alpha<\alpha^*$;

(iii) if $\overline{E}(Q,\alpha) \neq Q$, then the boundary of each maximal connected component $X$ of $\overline{E}$ is a single point $x$, which we call the {\bf local root} of $X$. When $x$ is removed, the remaining disjoint components of $X$ are subtrees of measure at most $\alpha/2$. We will also refer to $x$ as the local root of these subtrees.
\end{proposition}

We have labeled the local roots $x_1,x_2,\ldots$ in Figure~\ref{fig1}. Both the location and number of local roots may change as $\alpha$ changes. In the case $\alpha=2$, the set $\overline{E}$ has four maximal connected components, and four corresponding local roots, $x_1, x_2, x_3$ and $x_4$. When $\alpha=4$ or $6$, the set $\overline{E}$ has only three maximal connected components with local roots $x_1$, $x_2$ and $x_3$. When $\alpha=8$, the set $\overline{E}$ has only one maximal connected component. In this case, we have labeled the local root $x^*$, to be defined later in Subsection~\ref{sec:subtree}.

\cite{alpern2022continuous} showed that the $E$-patrolling strategy guarantees that the value of the continuous patrolling game on trees is at most $\alpha/(\mu+\lambda(E))$. Roughly speaking, the $E$-patrolling strategy repeatedly performs a tour of the tree, adding extra tours of each of the components of $\overline{E}$. Conjecture 1 of \cite{alpern2022continuous} was as follows.

\begin{conjecture}[Tree patrolling conjecture] \label{conj}
If $Q$ is a tree network, then for any~$\alpha$ the $E$-patrolling strategy is optimal and the value of the game is $v^* \equiv \alpha/(\mu+ \lambda(E))$.
\end{conjecture}

We will settle the tree patrolling conjecture in Section~\ref{sec:main}.

\section{Subtree Decomposition and Density} \label{sec:roots}

In this section we introduce the notion of the {\em local root of $Q$} and the {\em subtree decomposition} of a tree network in Subsection~\ref{sec:subtree} and the idea of {\em density} in Subsection~\ref{sec:density}.

\subsection{Subtree Decomposition} \label{sec:subtree}

In order to define the subtree decomposition of a tree network, we first introduce a new subset of $Q$ here called the {\em core} of $Q$, defined as the closure of the complement of $\overline{E}$ and denoted $E^0=E^0(Q,\alpha)$. The core is connected and closed. The reason for this rather awkward definition is that $E$ is only defined on regular points, but informally we can think of the core as the complement of the extremity set. The core is depicted in blue in Figure~\ref{fig1} for each value of $\alpha$. As $\alpha$ increases, the extremity set grows while the core shrinks. Notice that when $\alpha \geq 8$, the set $\overline{E}$ is equal to $Q$ and $E^0=\emptyset$.

Thus, for $\alpha < \alpha^*$, any tree network $Q$ can be expressed as the disjoint union of the core and a set of subtrees each of length at most $\alpha/2$ (see Proposition~\ref{prop:alpha}, part (iii)). This is the subtree decomposition of $Q$. It is easy to see that the core cannot contain any leaf nodes of $Q$.   In the remainder of this subsection we will show that for $\alpha \ge \alpha^*$, we can form a decomposition of $Q$ with similar properties.

If $\alpha \ge \alpha^*$, the set $\overline{E}$ has only one connected component, which is equal to $Q$. In this case, we define the {\em local root of $Q$}.
\begin{definition}
Let $Q$ be a tree and let $\alpha_1,\alpha_2, \ldots$ be a sequence of increasing positive numbers converging to $\alpha^*$. The local root of $Q$ is the set $ \cap_{n=1}^\infty E^0(Q,\alpha_n)$.
\end{definition}

It is easy to show that the local root of $Q$ is specified independently of the choice of sequence $(\alpha_n)_{n=1}^\infty$, and is in fact equal to $\cap_{0<\alpha < \alpha^*} E^0(Q,\alpha)$. The fact that the local root is non-empty follows from Cantor's intersection theorem, since it is the intersection of a sequence of non-empty, non-increasing, closed sets, by Proposition~\ref{prop:alpha}, part (i). In fact, we will show in Proposition~\ref{pro1} that the local root of $Q$ is a singleton, and without ambiguity, we will call its unique member the local root of $Q$ and denote it by $x^*$. The local root of the tree $Q$ is labeled in Figure~\ref{fig1}.

\begin{proposition}\label{pro1}
Let $Q$ be a tree. Then, 

(i) The local root of $Q$ is a singleton, $x^*$.

(ii) Each of the maximal connected components of $Q-\{x^*\}$ has measure at most $\alpha^*/2$.
\end{proposition}

\begin{proof} For (i), let $(\alpha_n)_{n=1}^\infty$ be an increasing sequence converging to $\alpha^*$ and let $f$ be the real function defined by $f(\alpha)=\lambda(E^0(Q,\alpha))$. Then $f$ is a continuous, and it follows that 
\[
\lambda(E^0(Q,\alpha_n))=f(\alpha_n) \rightarrow f(\alpha^*) = \lambda(E^0(Q,\alpha^*))=0.
\]
Now suppose the local root of $Q$ contains two points $x$ and $y$ with $x \neq y$, and let $\varepsilon=d(x,y)$. Let $N$ be such that $f(\alpha_N) < \varepsilon$. Since $E^0(Q,\alpha_N)$ is connected and contains both $x$ and $y$, it must contain the path from $x$ to $y$. Therefore, its measure must be at least $\varepsilon$, contradicting $f(\alpha_N) < \varepsilon$. So the local root of $Q$ is a singleton, $x^*$.

To prove (ii), assume for a contradiction that $Q-\{x^*\}$ has a component $Q_1$ with  $\lambda(Q_1) > \alpha^*/2$. First suppose that $x^*$ is a regular point. In this case, $Q-\{x^*\}$ only has two components, and by definition of $\alpha^*$ (Proposition~\ref{prop:alpha}, part (ii)), the other component $Q_2$ must satisfy $\lambda(Q_2) < \alpha^*/2$. Let $\alpha' = \lambda(Q_2)+ \alpha^*/2 < \alpha^*$. Since $\lambda(Q_2)<\alpha'/2$,  we must have $x^* \in E(Q,\alpha')$ by definition of the extremity set. But by definition of $x^*$ and because $\alpha'<\alpha^*$, we must have $x^* \in E^0(Q,\alpha')$, which is a contradiction, since $E(Q,\alpha') \cap E^0(Q,\alpha')=\emptyset$

Now suppose $x^*$ is a node and let $Q'$ be the subtree  $Q_1 \cup \{x^*\}$. Then,  $\lambda(Q') = \lambda(Q_1) > \alpha^*/2$. Let $z$ be a regular point on the arc incident to $x^*$ in $Q'$ such that $d(x^*, z) < \lambda(Q_1) -\alpha^*/2$.  It is easy to see that one component $Q_1(z)$ of $Q - \{z\}$ is a subset of $Q_1$ and the other component $Q_2(z)$ contains $x^*$.  We have $\lambda(Q_1(z)) = \lambda(Q_1) - d(x^*,z) > \alpha^*/2$. So, $\lambda(Q_2(z)) < \alpha^*/2$, by definition of $\alpha^*$. Let $\alpha'' = \lambda(Q_2(z)) + \alpha^*/2 < \alpha^*$. Since $\lambda(Q_2(z))<\alpha''/2$, we must have $x^* \in Q_2(z) \subset \overline{E}(Q,\alpha'')$ by definition of the extremity set. Because $x^*$ is not in the boundary of $Q_2(z)$, it is obvious that $x^* \not \in E^0(Q,\alpha'')$. But by definition of $x^*$ and because $\alpha''<\alpha^*$, we must have $x^* \in E^0(Q,\alpha'')$, which is again a contradiction.
\end{proof}

The local root $x^*$ of $Q$ is labeled on part (d) of  Figure~\ref{fig1}. Note that it does not depend on $\alpha$.

Proposition~\ref{pro1} implies that the tree $Q$ can be expressed as a disjoint union of its local root and a set of subtrees of length at most $\alpha/2$.  Combining this with the decomposition described earlier in this subsection for $\alpha < \alpha^*$, we have shown the following.

\begin{proposition}[Subtree decomposition of $Q$]
For any tree $Q$ and any attack time $\alpha$, we can express $Q$ as a union of its core $E^0$ and a set of closed subtrees $E^1,\ldots,E^k$ whose union is $\overline{E}$ such that $\lambda(E^i) \le \alpha/2$ for each $i=1,\ldots,k$ and $\sum_{i=1}^k \lambda(E^i) = \lambda(E)$.
\end{proposition}

\subsection{Density}
\label{sec:density}

In this subsection we introduce the concept of {\em density}. 

Suppose a measure $P$ on $Q$ is fixed. For any measurable $A \subseteq Q$, we define the \textit{density} $\rho_P(A)\equiv \rho(A)$ by $\rho(A)=P(A)/\lambda(A)$. 

Suppose $Q$ is a tree with a distinguished point $O$, called its {\em root}. We say a point $y \in Q$ is {\em above} a point (or arc) $x$ if the unique path from $O$ to $y$ contains $x$. We write $Q_x$ for the subtree of $Q$ containing $x$ and all points above $x$. We call a node $x$ a {\em branch node} if it is not a leaf node. For a branch node $x$ of $Q$, we call the {\em branches at $x$} the collection of maximal disjoint components of $Q_x-\{x\}$. 

We state the definition of the \textit{Equal Branch Density} (EBD) distribution, as given in \cite{alpern2013search}, \cite{alpern10} and \cite{alpern14}.
\begin{definition}
For a tree $Q$ with root $O$, the \textit{Equal Branch Density} (EBD) distribution is the unique measure~$h$ on the leaf nodes of $Q$ (not including $O$) such that at every branch node $x$ all the branches at $x$ have the same density~$\rho_h$.
\end{definition}


We state here an important property of the EBD distribution, which is a consequence of Lemma~6 of \cite{alpern2013search}.

\begin{lemma}\label{lem:EBD} The EBD distribution $h$ on a rooted tree $Q$ has the property that for any subtree $Z$ with root $x$ contained in $Q_x$, we have $\rho_h(Z) \leq \rho_h(Q_x)$. 
\end{lemma}

\section{Proof of the Tree Patrolling Conjecture} \label{sec:main}

We begin this section by constructing an Attacker strategy in Subsection~\ref{sec:attack}, which we call the {\em tree attack strategy}. In Subsection~\ref{sec:proof}, we will show that this strategy is $\varepsilon$-optimal.

\subsection{The Tree Attack Strategy} \label{sec:attack}

The tree attack strategy is actually a collection of strategies, and is defined in terms of a parameter $T >0$, which we can think of as the length of some long time interval.

\begin{definition}[tree attack strategy]
Let $Q$ be a tree network, and let $E^0, E^1,\ldots, E^k$ be its subtree decomposition. Let $x_j$ be the local root of $E^j$ for $j=1,..,k$. Let $h^j$ be the EBD measure on $E^j$. For $T>0$, the tree attack strategy (with parameter $T$) begins at a time chosen uniformly at random from the interval $[0,T]$. The location of the attack is given by the measure $e$, defined below.

(i) With probability $e(E^0) \equiv \lambda(E^0)/(\mu+\lambda(E))$, a point of $E^0$ chosen uniformly at random. 

(ii) With probability $e(E^j) \equiv 2\lambda(E^j)/(\mu+\lambda(E))$, a point of $E^j$ chosen according to the EBD distribution~$h^j$, for $j=1,\ldots,k$.
\end{definition}

The tree attack strategy is well defined. Indeed, the total probability $e(Q)$ of attack is given by
$$e(Q) = \sum_{j=0}^k e(E^j) = \frac{\lambda(E^0)}{\mu+\lambda(E)}+\sum_{j=1}^k\frac{2\lambda({E^j})}{\mu+\lambda(E)}=\frac{\lambda(E^0)}{\mu+\lambda(E)}+\frac{2\lambda({E})}{\mu+\lambda(E)}=\frac{\mu+\lambda(E)}{\mu+\lambda(E)}=1.$$
 
We illustrate the tree attack strategy by revisiting the network $Q$ with length $\mu =10$ from Figure~\ref{fig1}. We illustrate the attack probability at the leaf nodes and in $E^0$ in Figure~\ref{fig2} for different values of~$\alpha$.
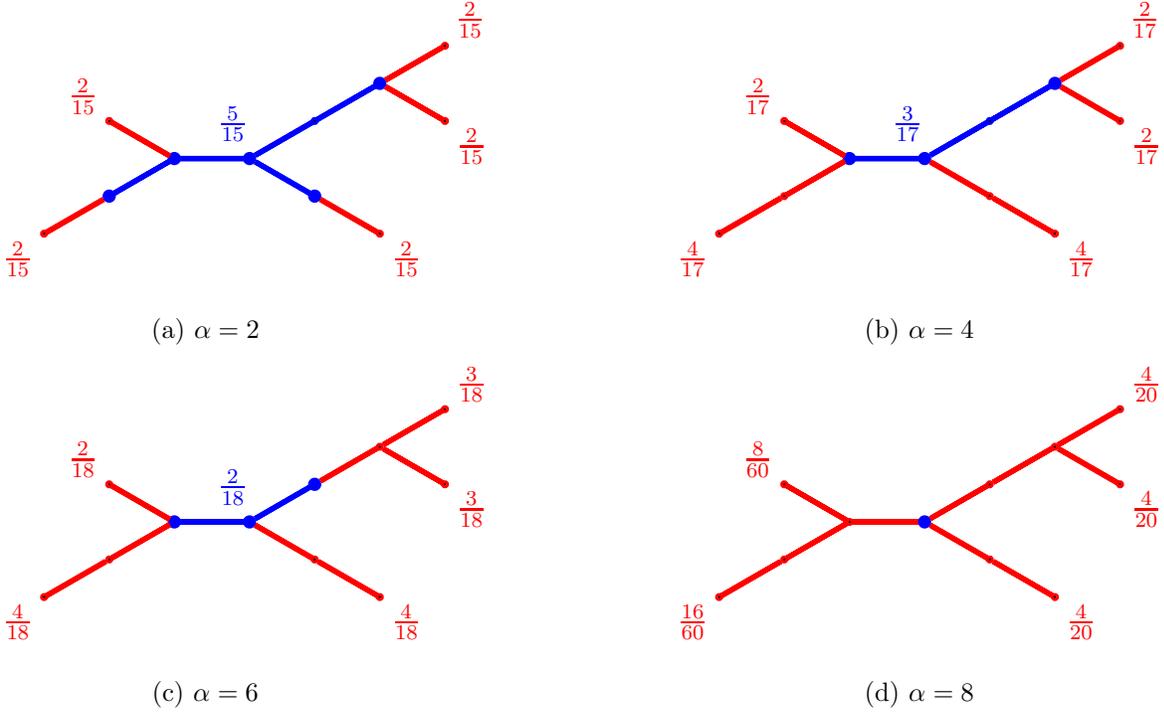
\begin{figure}[ht]
    \centering
    \begin{subfigure}[b]{0.35\textwidth}
         \centering
        \begin{tikzpicture}[line width=1pt]
    \tikzstyle{every node}=[draw,circle,fill=black,minimum size=1pt,
                            inner sep=0pt]
    \draw (0,0) [color=blue] node (a)[label=10:${\ \ \ \ \ \frac{5}{15}\ }$, color=blue, minimum size=3pt] {}
        -- ++(0:1cm) [color= blue, line width=2pt] node (0)[ color=blue, minimum size=3pt]  {}
        -- (a);
    \draw (0)
       -- ++(30:1cm) [color= blue, line width=2pt] node (12) {}
       -- ++(30:1cm) node (1) [ color=blue, minimum size=3pt]{}
       -- (0);
    \draw (0)
       -- ++(-30:1cm) [color= blue, line width=2pt] node (2) [color=blue, minimum size=3pt]{}
       -- (0);
    \draw (2)
       -- ++(-30:1cm) [color= red, line width=2pt] node (22) [label=-30:${\ \frac{2}{15}\ }$]{}
       -- (2);
    \draw (1)
       -- ++(30:1cm) [color= red, line width=2pt] node (3) [label=60:${\ \frac{2}{15}\ }$]{}
       -- (1);
    \draw (1)
       -- ++(-30:1cm) [color= red, line width=2pt] node (4) [label=-30:${\ \frac{2}{15}\ }$]{}
       -- (1);
    \draw (a)
     -- ++(150:1cm)[color= red, line width=2pt] node (5) [label=150:${\ \frac{2}{15}\ }$]{}
     -- (a);
    \draw (a)
     -- ++(-150:1cm)[color= blue, line width=2pt] node (6)[ color=blue, minimum size=3pt] {}
     -- (a);
    \draw(6)
     -- ++(-150:1cm)[color= red, line width=2pt] node (62) [label=-150:${\ \frac{2}{15}\ }$]{}
     -- (6);
\end{tikzpicture}
         \caption{$\alpha=2$}
         \label{fig:1}
     \end{subfigure}
    \hfill
    \begin{subfigure}[b]{0.5\textwidth}
         \centering
    \begin{tikzpicture}[line width=1pt]
    \tikzstyle{every node}=[draw,circle,fill=black,minimum size=1pt,
                            inner sep=0pt]
    \draw (0,0) [color=blue] node (a)[label=10:${\ \ \ \ \ \frac{3}{17}\ }$, color=blue, minimum size=3pt] {}
        -- ++(0:1cm) [color= blue, line width=2pt] node (0) [color=blue, minimum size=3pt] {}
        -- (a);
    \draw (0)
       -- ++(30:1cm) [color= blue, line width=2pt] node (12) {}
       -- ++(30:1cm) node (1) [color=blue, minimum size=3pt]{}
       -- (0);
    \draw (0)
       -- ++(-30:1cm) [color= red, line width=2pt] node (2) {}
       -- (0);
    \draw (2)
       -- ++(-30:1cm) [color= red, line width=2pt] node (22) [label=-30:${\ \frac{4}{17}\ }$]{}
       -- (2);
    \draw (1)
       -- ++(30:1cm) [color= red, line width=2pt] node (3) [label=60:${\ \frac{2}{17}\ }$]{}
       -- (1);
    \draw (1)
       -- ++(-30:1cm) [color= red, line width=2pt] node (4) [label=-30:${\ \frac{2}{17}\ }$]{}
       -- (1);
    \draw (a)
     -- ++(150:1cm)[color= red, line width=2pt] node (5) [label=150:${\ \frac{2}{17}\ }$]{}
     -- (a);
    \draw (a)
     -- ++(-150:1cm) [color= red, line width=2pt] node (6) {}
     -- (a);
    \draw(6)
     -- ++(-150:1cm)[color= red, line width=2pt] node (62) [label=-150:${\ \frac{4}{17}\ }$]{}
     -- (6);
\end{tikzpicture}
         \caption{$\alpha=4$}
         \label{fig:2}
     \end{subfigure}
     
    \begin{subfigure}[b]{0.35\textwidth}
         \centering
        \begin{tikzpicture}[line width=1pt]
    \tikzstyle{every node}=[draw,circle,fill=black,minimum size=1pt,
                            inner sep=0pt]
    \draw (0,0) [color= blue] node (a)[label=10:${\ \ \ \ \ \frac{2}{18}\ }$, color=blue, minimum size=3pt] {}
        -- ++(0:1cm) [color= blue, line width=2pt] node (0) [color=blue, minimum size=3pt] {}
        -- (a);
    \draw (0)
       -- ++(30:1cm) [color= blue, line width=2pt]node (1) [color=blue, minimum size=3pt]{}
       -- (0);
    \draw (1)
       -- ++(30:1cm) [color= red, line width=2pt] node (12) {}
       -- (1);
    \draw (0)
       -- ++(-30:1cm) [color= red, line width=2pt] node (2) {}
       -- (0);
    \draw (2)
       -- ++(-30:1cm) [color= red, line width=2pt] node (22) [label=-30:${\ \frac{4}{18}\ }$]{}
       -- (2);
    \draw (12)
       -- ++(30:1cm) [color= red, line width=2pt] node (3) [label=30:${\ \frac{3}{18}\ }$]{}
       -- (12);
    \draw (12)
       -- ++(-30:1cm) [color= red, line width=2pt] node (4)[label=-30:${\ \frac{3}{18}\ }$]{}
       -- (12);
    \draw (a)
     -- ++(150:1cm)[color= red, line width=2pt] node (5) [label=150:${\ \frac{2}{18}\ }$]{}
     -- (a);
    \draw (a)
     -- ++(-150:1cm) [color= red, line width=2pt] node (6) {}
     -- (a);
    \draw(6)
     -- ++(-150:1cm)[color= red, line width=2pt] node (62) [label=-150:${\ \frac{4}{18}\ }$]{}
     -- (6);
\end{tikzpicture}
         \caption{$\alpha=6$}
         \label{fig:3}
     \end{subfigure}
    \hfill
    \begin{subfigure}[b]{0.5\textwidth}
         \centering
         \begin{tikzpicture}[line width=1pt]
    \tikzstyle{every node}=[draw,circle,fill=black,minimum size=1pt,
                            inner sep=0pt]
    \draw (0,0) [color= blue] node (a) {}
        -- ++(0:1cm) [color= red, line width=2pt] node (0) [color=blue, minimum size=3pt] {}
        -- (a);
    \draw (0)
       -- ++(30:1cm) [color= red, line width=2pt] node (1) {}
       -- (0);
    \draw (1)
       -- ++(30:1cm) [color= red, line width=2pt] node (12) {}
       -- (1);
    \draw (0)
       -- ++(-30:1cm) [color= red, line width=2pt] node (2) {}
       -- (0);
    \draw (2)
       -- ++(-30:1cm) [color= red, line width=2pt] node (22) [label=-30:${\ \frac{4}{20}\ }$]{}
       -- (2);
    \draw (12)
       -- ++(30:1cm) [color= red, line width=2pt] node (3) [label=30:${\ \frac{4}{20}\ }$]{}
       -- (12);
    \draw (12)
       -- ++(-30:1cm) [color= red, line width=2pt] node (4) [label=-30:${\ \frac{4}{20}\ }$]{}
       -- (12);
    \draw (a)
     -- ++(150:1cm)[color= red, line width=2pt] node (5) [label=150:${\ \frac{8}{60}\ }$]{}
     -- (a);
    \draw (a)
     -- ++(-150:1cm) [color= red, line width=2pt] node (6) {}
     -- (a);
    \draw(6)
     -- ++(-150:1cm)[color= red, line width=2pt] node (62) [label=-150:${\ \frac{16}{60}\ }$]{}
     -- (6);
\end{tikzpicture}
         \caption{$\alpha=8$}
         \label{fig:4}
     \end{subfigure}
    \caption{The tree attack strategy on the tree network $Q$.}
    \label{fig2}
\end{figure}

Observe that the density $\rho_e(E^j)\equiv \rho(E^j)$ for any $j=1,\ldots,k$ is  $$\rho(E^j)=\frac{e(E^j)}{\lambda(E^j)}=\frac{2\lambda({E^j})}{\mu+\lambda(E)}\frac{1}{\lambda(E^j)}=\frac{2}{\mu+\lambda(E)}. $$

So, by Lemma~\ref{lem:EBD}, for any subtree $Z$ of  $E^j$ such that $x_j \in Z$, 
\begin{align} 
\rho(Z) \leq \rho(E^j)=\frac{2}{\mu+\lambda(E)}. \label{eq:rho}
\end{align}

\subsection{$\varepsilon$-Optimality of the Tree Attack Strategy} \label{sec:proof}

Before proving the tree patrolling conjecture, we extend a lemma from~\cite{alpern2022continuous} concerning the {\em uniform attack strategy}. This is the strategy for the attacker that begins the attack at an arbitrary time $M$ (for example $M=0$) at a point of the network chosen uniformly at random. \cite{alpern2022continuous} showed that this strategy ensures the attack will be intercepted with probability at most $\alpha/\mu$ (this was also shown in \cite{alpern2016patrolling} and \cite{garrec2019continuous}).

\begin{lemma} \label{lem:unif}
Let $Z$ be a connected subset of a network $Q$. Consider an attack strategy that chooses a point of $Z$ uniformly at random to carry out the attack, and starts the attack at some time $t$, which may be fixed or a random variable. Then for any Patroller strategy, the probability that attack is intercepted is at most $\alpha/\lambda(Z)$.
\end{lemma}
\begin{proof}
The lemma is trivially true if $\alpha \ge \lambda(Z)$, so assume that $\alpha < \lambda(Z)$. First suppose $t$ is fixed. Then Lemma~1 of \cite{alpern2022continuous} applied to the network $Z$ says that probability of interception is at most $\alpha/\lambda(Z)$. 

Now suppose $t$ is a random variable. Then from the previous paragraph, the probability the attack is intercepted, conditional on the attack starting at fixed time $t=t_0$ is at most $\alpha/\lambda(Z)$. It follows that the unconditional probability of interception is also at most $\alpha/\lambda(Z)$. 
\end{proof}

We are now ready to prove the tree patrolling conjecture.

\begin{theorem}
Let $Q$ be a tree of length $\mu$. Then for any $\varepsilon >0$, there exists a value of $T$ such that the tree attack strategy (with parameter $T$) cannot be intercepted with probability greater than $\alpha/(\mu+\lambda(E))+\varepsilon \equiv v^*+\varepsilon$. 
Hence, the value of the continuous patrolling game on $Q$ is $v^* $ and the $E$-patrolling strategy is optimal.
\end{theorem}
\begin{proof} Let $\varepsilon>0$ be given, and suppose the Attacker uses the tree attack strategy (with parameter $T$), for some $T$, where the precise value of $T$ will be specified later. Consider an arbitrary patrol $S$, and let $0=t_0 < t_1 <\cdots <t_m=T+\alpha$ be the coursest partition of $[0,T+\alpha]$ such that $S$ is confined to a single set $E^j$ ($j=0,1,\ldots,k$) during each time interval $[t_i,t_{i+1}]$. For $i=1,\ldots,m$, let $I_i=[t_{i-1},t_i]$, let $\delta_i = t_i-t_{i-1}$ and let $Z_i=S(I_i)$.

We will show that the probability $P(S)$ that $S$ intercepts the tree attack strategy is at most $v^*+\varepsilon$. To do so, we will calculate an upper bound for the probability $P_i$ that $S$ intercepts the attack during each of the intervals $I_i$ for each $i=1,\ldots,m$, and we will show that the sum of these upper bounds is no more than $v^*+ \varepsilon$. 

First suppose $m=1$. In this case, the patrol just stays in one component $E^j$ during the whole time $[0, T+\alpha] \equiv I_1.$ If $j\neq0$, the interception probability $P_i$ is at most $$e(E^j)=2\lambda(E^j)/(\mu+\lambda(E))\leq \alpha/(\mu+\lambda(E))=v^*,$$ since $\lambda(E^j) \le \alpha/2$. If $j=0$, then by Lemma~\ref{lem:unif}, then the interception probability $P_i$ satisfies 
$$P_i \le \frac{\alpha}{\lambda(E^0)}\cdot e(E^0) =\frac{\alpha}{\lambda(E^0)} \cdot \frac{\lambda(E^0)}{\mu+\lambda(E)}=v^*.$$

Now suppose $m\geq 2$, and we calculate an upper bound of interception probability $P_i$ in three cases: 
\begin{enumerate}
   \item[(i)] $Z_i \subseteq E^0$;
   \item[(ii)]  $i=2,.., m-1$ and $Z_i \subseteq E^j$ for some $j=1,\ldots,k$;
   \item[(iii)] $i=1 \text{ or } m$ and $Z_i \subseteq E^j$ for some $j=1,\ldots,k$. 
\end{enumerate}

Starting with case (i), when $Z_i \subseteq E^0$, the interception probability $P_i$ is no greater than the product of the probability the attack starts in the interval $I_i$, the probability the attack takes place in $E^0$ and the conditional probability that $S$ intercepts the attack given that it takes place in $E^0$ starting during $I_i$. Using Lemma~\ref{lem:unif}, this gives the bound
\begin{align}\label{eq1}
   P_i \le   \frac{\delta_i}{T}  \cdot e(E^0) \cdot  \frac{ \alpha }{\lambda(E^0)} = \frac{\delta_i v^* }{T}.
\end{align}

Second, in the case that $i=2,.., m-1$ and $Z_i \subseteq E^j$ for some $j=1,\ldots,k$, the patrol must perform a tour with the startpoint and endpoint $x_j$. Because the length of this tour is at least $2\lambda(Z_i)$, the patrol can spend at most time $\delta_i- 2\lambda(Z_i) \geq 0$ at leaf nodes of $E^j$. Therefore, $P_i$ satisfies
$$P_i \le \frac{1}{T} e(Z_i)(\alpha+\delta_i-2\lambda(Z_i)).$$
By~(\ref{eq:rho}), $\rho(Z_i) = e(Z_i)/\lambda(Z_i) \leq \rho(E^j)=2/(\mu+\lambda(E))$. Applying this to the inequality above and rearranging, 
\begin{align*}
P_i & \le \frac{1}{T} e(Z_i)(\alpha+\delta_i-2\lambda(Z_i)) \\
&\leq \frac{1}{T} \frac{2\lambda(Z_i)}{\mu+\lambda(E)}(\alpha+\delta_i-2\lambda(Z_i)) \\
&=\frac{1}{T}\frac{\alpha}{\mu+\lambda(E)}\left(\delta_i- \big (\delta_i-2\lambda(Z_i)\big)\left(1-\frac{2\lambda(Z_i)}{\alpha} \right)\right).
\end{align*}
As already observed, $\delta_i-2\lambda(Z_i) \geq 0$. Also, $\lambda(Z_i) \leq \lambda (E^j) \leq \alpha/2$, by definition of the subtree decomposition, so $\big (\delta_i-2\lambda(Z_i)\big)\big(1-\frac{2\lambda(Z_i)}{\alpha} \big) \geq 0$. Consequently,
\begin{equation}\label{eq2}
    P_i \leq \frac{1}{T}\frac{\alpha}{\mu+\lambda(E)}\delta_i = \frac{\delta_i v^* }{T}.
\end{equation}

Third, we consider the case that $i=1 \text{ or } m$ and $Z_i \subseteq E^j$ for some $j=1,\ldots,k$. This case is different from the second case since it is not necessary for the patrol to perform a tour in $E^j$. For example, the patrol may start at a leaf node in $Z_1$, stay within $Z_1$ for sometime then move directly to $Z_2$. Therefore, the time the patrol can stay at leaf nodes in $E^j$ is at most $\delta_i-\lambda(Z_i) \geq 0$, and the interception probability $P_i$ satisfies
$$P_i \le \frac{1}{T} e(Z_i)(\alpha+\delta_i-\lambda(Z_i)).$$

The condition $\rho(Z_i) \leq \alpha/(\mu+\lambda(E))$ still holds since $Z_i$ contains $x_j$, and must therefore be a subtree of $E^j$. Applying this to the inequality above and rearranging,
\begin{align*}
P_i &\leq \frac{1}{T}\frac{2\lambda(Z_i)}{\mu+\lambda(E)}(\alpha+\delta_i-\lambda(Z_i)) \\
&=\frac{1}{T}\frac{\alpha}{\mu+\lambda(E)}\left(\delta_i+2\lambda(Z_i)-\left(1-\frac{2\lambda(Z_i)}{\alpha}\right)\delta_i-\frac{2(\lambda(Z_i))^2}{\alpha}\right).
\end{align*}

Since $\lambda(Z_i) \leq \lambda (E^j) \leq \alpha/2$, we have $(1-2\lambda(Z_i)/\alpha)\delta_i\geq 0$ and
\begin{equation}\label{eq3}
   P_i \leq \frac{1}{T}\frac{\alpha}{\mu+\lambda(E)}\big(\delta_i+2\lambda(Z_i)\big)\leq\frac{1}{T}\frac{\alpha}{\mu+\lambda(E)}(\delta_i+\alpha)= \frac{(\delta_i+\alpha) v^* }{T}.
\end{equation}

Combining inequalities (\ref{eq1}) - (\ref{eq3}), we obtain
$$P(S) \le \sum_{i=1}^m P_i \le \frac{2\alpha v^* }{T} +\sum_{i=1}^m \frac{\delta_i v^* }{T} =  \frac{2\alpha v^* }{T} + \frac{v^* }{T}(T+\alpha) \leq v^*+\varepsilon,$$
where we choose $T=3\alpha/\varepsilon$.

We have shown that the tree attack strategy cannot be intercepted with probability greater than $v^*+\varepsilon$, so that the value of the game is at most $v^*+\varepsilon$. Combining this with the lower bound of $v^*$ from \cite{alpern2022continuous} given by the $E$-patrolling strategy, the rest of the theorem follows. 
\end{proof}

\section{Solving the Game for Complete Networks} \label{sec:complete}
In this section, we study the game on complete networks. We begin this section by introducing some standard definitions and the concept of a $k$-factorization of complete networks in Subsection 5.1. In Subsection 5.2, we introduce a Patroller strategy which we call the \textit{complete network patrolling strategy} and show that this strategy is optimal for some values of $\alpha$.  
\subsection{$k$-factorization of Complete Networks}\label{subsec:1factorization}
In this section, we just consider simple networks (i.e networks that do not contain any loops and for which there is at most one arc connecting any pair of nodes). A \textit{k-regular network} is a simple network all of whose nodes have degree $k$ ($k \geq 1$). A \textit{complete network} is a $k$-regular network on $m$ ($m\geq 2$) nodes where $k = m-1$. We denote a complete network with $n$ nodes by $K_n$. Note that since the arcs of $K_n$ may have different lengths, it is not uniquely defined.

We denote the set of arcs of a network $Q$ by $E(Q)$ and the set of nodes of $Q$ by $V(Q)$. 
\begin{definition} Let $Q$ be a $k$-regular network. A \textit{k-factorization} of $Q$ is a set of sub-networks $F=\{F_1,\ldots,F_z\}$ such that 
\begin{enumerate}
\item[(i)] for all $i=1,\ldots,z$, the sub-network $F_i$ is a $k$-regular network with $V(F_i)=V(Q)$,
\item[(ii)] $E(Q)=\cup_{i=1}^z E(F_i)$ and 
\item[(iii)] for any $1 \leq i\neq j \leq z$, we have $E(F_i) \cap E(F_j) =\emptyset$.
\end{enumerate}
\end{definition}

In particular, a 1-factorization of $Q$ is a set of arc-disjoint perfect matchings whose union is $E(Q)$. In other words, a 1-factorization is an arc-coloring of a network where each color class consists of a perfect matching. In Figure~\ref{fig:factorization}, we illustrate a 1-factorization of a complete network $K_4$ on four nodes with three color classes, red, blue, and green.  
\begin{figure}[!ht]
    \centering
    \begin{tikzpicture}[line width=1pt]
    \tikzstyle{every node}=[draw,circle,fill=black,minimum size=1pt,
                            inner sep=0pt]
    \draw (0,0)[color = red] node (v4){}
        -- ++(90:2cm) node (v1){}
        -- (v4);
    \draw (v4)[color = blue]
       -- ++(-30:2cm) node (v3){}
       -- (v4);
    \draw (v4)[color = green]
       -- ++(-150:2cm) node (v2){}
       -- (v4);
    \draw [color = red](v2) -- (v3);
    \draw [color = blue](v1) -- (v2);
    \draw [color = green](v1) -- (v3);
\end{tikzpicture}
\caption{1-factorization of a complete network $K_4$.}
    \label{fig:factorization}
\end{figure}
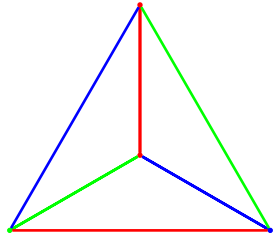
 
It is obvious that if $Q$ has a 1-factorization, the number of nodes of $Q$ must be even. It is well-known that every complete network on $2n$ ($n=1,\ldots$) nodes admits a 1-factorization. \cite{csaba} showed that every $k$-regular network on $2n$ nodes has a 1-factorization if $k \geq 2\lceil n/2 \rceil -1$. For a small number of nodes $2n\leq 4$, the network $K_{2n}$ has a unique 1-factorization. When $2n \geq 6$, the network $K_{2n}$ has many 1-factorizations \citep[see][]{zino2014}. For example, $K_8$ has 6240 distinct 1-factorizations. 

\subsection{A Patrolling Strategy for Complete Networks} \label{subsec:complete}
In this subsection, we introduce a Patroller strategy for the complete network $K_{2n}$ on $2n$ ($n=1,2,\ldots$) nodes and prove this strategy is optimal for some values of $\alpha$. Note that a complete network on an odd number of nodes is Eulerian. The solution for Eulerian networks was presented in \cite{garrec2019continuous} and \cite{alpern2022continuous}.

Suppose the complete network $K_{2n}$ has a 1-factorization $F=\{F_1,\ldots,F_{2n-1}\}$. We first observe that for any $i=1,\ldots, 2n-1$, the sub-network $Q_i= K_{2n}- F_i$ is a $k$-regular network where $k= 2n-2$. Therefore, $Q_i$ is Eulerian for all $i$. We define the complete network patrolling strategy below. 

\begin{definition}[complete network patrolling strategy]
Let $F=\{F_1,\ldots,F_{2n-1}\}$ be a 1-factorization of a complete network $K_{2n}$. For $i=1,\ldots,2n-1$, let $Q_i= K_{2n}- F_i$ and let $S_i$ be an Eulerian tour of $Q_i$ starting at a randomly chosen point. The complete network patrolling strategy $S^F$ is a patrol such that the Patroller chooses $S_i$ with probability $s_i= \lambda(Q_i)/((2n-2)\mu)$.
\end{definition}

Note that $S^F$ is well defined, since
$$\sum_{i=1}^{2n-1} s_i = \frac{\sum_{i=1}^{2n-1} \lambda(Q_i)}{(2n-2)\mu}=\frac{\sum_{i=1}^{2n-1} \big ( \lambda(K_{2n})- \lambda(F_i) \big )}{(2n-2)\mu}=\frac{\sum_{i=1}^{2n-1} \lambda(K_{2n})- \sum_{i=1}^{2n-1} \lambda(F_i)}{(2n-2)\mu}=\frac{(2n-1)\mu-\mu}{(2n-2)\mu}=1. $$

For a $1$-factorization $F$, let $\delta(F) = \max_{1 \leq i \leq 2n-1} \lambda(F_i)$. We have the following result.
\begin{proposition}\label{lem:complete}
    Let $F$ be a 1-factorization of the complete network $K_{2n}$ for some $n=1,2,\ldots$. For $\alpha \leq \mu - \delta(F)$, the strategy $S^F$ is optimal for the Patroller and the uniform attack strategy is optimal for the Attacker on $K_{2n}$. The value of the game is $V= \alpha/\mu$.
\end{proposition}
\proof
Consider an arbitrary attack taking place at some point $x \in Q$. We will show that the patrol $S^F$ can intercept this attack with probability at least $\alpha/\mu$. 

Let $P(S_i)$ be the probability that $S_i$ intercepts the attack, for $i=1,\ldots,2n-1$. For each $i$, we have $\lambda(Q_i)=\mu-\lambda(F_i) \geq \mu -\delta(F) \geq \alpha$. It follows from Corollary 1 of \cite{alpern2022continuous} that $P(S_i)=\alpha/\lambda(Q_i)$.

If $x$ is a node, it is easy to see that $x \in Q_i$ for all $i$. So, the patrol $S^F$ will intercept the attack with probability $$P(S^F) = \sum_{i=1}^{2n-1} s_i P(S_i)= \sum_{i=1}^{2n-1} \frac{\lambda(Q_i)}{(2n-2)\mu}\frac{\alpha}{\lambda(Q_i)}= \frac{2n-1}{2n-2}\frac{\alpha}{\mu} \geq \frac{\alpha}{\mu}.$$

If $x$ is not a node, there exists a unique $j \in [2n-1]$ such that $x \in F_j$ (where $[m]$ denotes the set $\{1,\ldots,m\}$). So, $x \not \in Q_j$ and $x \in Q_i$ for all $i \neq j$. The probability the patrol $S^F$ intercepts the attack is $$P(S^{F}) = \sum_{\substack{i=1}}^{2n-1} s_i P(S_i)= \sum_{\substack{i\in [2n-1] \\ i\neq j}} \frac{\lambda(Q_i)}{(2n-2)\mu}\frac{\alpha}{\lambda(Q_i)}= \frac{2n-2}{2n-2}\frac{\alpha}{\mu} = \frac{\alpha}{\mu}.$$ 

We have shown that $V \ge \alpha/\mu$. But by Lemma 1 of \cite{alpern2022continuous} (or Lemma~\ref{lem:unif} of this paper), the uniform attack strategy guarantees that $V \le \alpha/\mu$ for any network. We conclude that the strategy $S^F$ and the uniform attack strategy are optimal and the value of the game is $V= \alpha/\mu$.  
\endproof

As mentioned in Subsection \ref{subsec:1factorization}, when the number of nodes $2n \geq 6$, the network $K_{2n}$ has many 1-factorizations. Let $F^*$ be a 1-factorization of $K_{2n}$ such that $\delta^*= \delta(F^*) \leq \delta(F)$ for any 1-factorization $F$ of $Q$. We then have the following stronger result.
\begin{proposition}
For $\alpha \leq \mu - \delta^*$, the strategy $S^{F^*}$ and the uniform attack strategy are optimal. The value of the game is $V= \alpha/\mu$.
\end{proposition}

In comparison with the recent work of \cite{alpern2022continuous}, the complete network patrolling strategy $S^{F}$ helps us solve the game for a significantly larger range of $\alpha$. \cite{alpern2022continuous} introduced a patrolling strategy for networks without leaf arcs and proved it is optimal for $\alpha \leq g$ where $g$ is the girth of the network, defined as the minimum length of a circuit in the network. For a complete network $K_{2n}$ ($n \geq 2$), the girth $g$ is very small compared to $\mu - \delta(F)$, for any 1-factorization $F$. In fact, we have $\mu - \delta(F) \geq \frac{n(n-1)}{2}g$. Indeed, by Theorem 1 of \cite{alspach2001decompostion}, any sub-network $Q_i = K_{2n} - F_i$ ($F_i \in F$) can be decomposed into circuits $C_4$ of four arcs. Since $|E(Q_i)| =2n(n-1),$ a $C_4$-decomposition of $Q_i$ has $n(n-1)/2$ circuits and the length of any circuit is not less than $g$ by definition of $g$. So, $\lambda(Q_i) \geq \frac{n(n-1)}{2}g$. Therefore, $\mu- \delta(F) \geq \min_i \lambda(Q_i) \geq \frac{n(n-1)}{2}g$ for any 1-factorization $F$.

In summary, the patrolling strategy of \cite{alpern2022continuous} is known to be optimal for values of $\alpha$ in $(0,g]$, whereas the complete network patrolling strategy is known to be optimal for values of $\alpha$ in $(0,\frac{n(n-1)}{2}g]$, an interval that is $O(n^2)$ longer.

Notice that if $Q$ is a network with unit length arcs (i.e every arc is of length 1), then $\delta(F) = n$ for all 1-factorizations $F$, so that $\delta^* = n$. Thus, for $\alpha \leq \mu - n$, the value of the game is $\alpha/\mu$. In fact, this bound can be tight. In other words, for some networks, for $\alpha> \mu - \delta^*$, the value of the game is strictly less than $\alpha/\mu$. 

\begin{proposition} \label{prop:tightness} Consider an attack strategy for the network $K_4$ with unit length arcs which attacks at a random point with a start time chosen uniformly at random from the interval $[0, 6-\alpha]$. For $\mu - n=4 <\alpha \leq 6$, this attack strategy guaranees an interception probability of stricty less than $\alpha/\mu$. 
\end{proposition}

The proof of Proposition~\ref{prop:tightness} is in the Appendix.

\begin{remark}
Proposition~\ref{lem:complete} can be extended to general $k$-regular networks $Q$ on $2n$ ($n \geq 2$) nodes. We consider the case $k \geq n+1$ and $k$ is odd. From Subsection \ref{subsec:1factorization}, we know $Q$ admits a 1-factorization $F=\{F_1,\ldots,F_{k}\}$. Also, for all $i=1,\ldots,k$, the sub-network $Q_i=Q-F_i$ is Eulerian. It is well known that a network $Q$ all of whose nodes have degree at least $|V(Q)|/2$ is connected. So $Q_i$ must be connected because all its nodes are of degree $k-1 \geq n$. Let $S_i$ be an Eulerian tour of $Q_i$ which starts at a random point. Let $S$ be a patrolling strategy which chooses $S_i$ with probability $\lambda(Q_i)/\sum_{j=1}^{k}\lambda(Q_j)$. Similarly to the proof of Proposition \ref{lem:complete}, it is easy to show that the strategy $S$ is optimal and the value of the game is $\alpha/\mu$ for $\alpha \leq \mu - \delta(F)$. 
\end{remark}

In general, if we know a $k$-regular network has an $m$-factorization, we can generalize Proposition~\ref{lem:complete} as follows. 

\begin{theorem} \label{thm:mfact}
Let $Q$ be a $k$-regular network on $2n$ vertices such that $n\geq 2$ and $k$ is odd. Assume $Q$ admits an $m$- factorization $F_m$ for some odd $m$ such that $k\geq n+m$. Then, the value of the game is $V=\alpha/\mu$ for $\alpha \leq \mu -\delta(F)$. 
\end{theorem}
\proof
Observe that $|F_m| = k/m= r$. Let $F_m= \{F_1,\ldots, F_r\}$ and $Q_i= Q- F_i$ for $F_i \in F$. Then, for all $i= 1,\ldots, r$, the sub-network $Q_i$ is a $k'$-regular network where $k'= k-m$. Since $k \geq n+m$, we have $k' \geq n = V(Q_i)/2$ and $Q_i$ is connected. Moreover, $k'$ is even because $k$ and $m$ are odd. Therefore, $Q_i$ is Eulerian for all $i= 1, \ldots, r$. 

Let $S_i$ be an Eulerian tour of $Q_i$ which starts at a point chosen randomly. Let $S^{F_m}$ be a patrolling strategy which picks $S_i$ with probability $s_i = \lambda(S_i)/\sum_{j=1}^r \lambda(S_j)$. 
Then, similarly to the proof of Proposition~\ref{lem:complete}, it can be shown that for $\alpha \leq \mu - \delta(F) \leq \min_i \lambda(Q_i)$, the patrol $S^{F_m}$ can intercept any attack with probability at least $\alpha/\mu$ and $V \geq \alpha/\mu$. Since the uniform attack strategy can guarantee $V \leq \alpha/\mu$ \citep{alpern2022continuous}, for $\alpha \leq \mu - \delta(F)$, we conclude the value of the game is $V=\alpha/\mu$, the uniform attack strategy is optimal for the Attacker and the patrol $S^{F_m}$ is optimal for the Patroller.
\endproof

\section{Conclusion} \label{sec:concl}
We have settled a conjecture posed by \cite{alpern2022continuous} and thus shown that for tree networks, an easily implementable patrolling strategy is optimal in the continuous patrolling game. Although we have found $\varepsilon$-optimal attack strategies, we believe that optimal attack strategies exist in all cases, and it may be of interest to refine the tree attack strategy defined in this paper to obtain precisely optimal strategies. 

We have also solved the game for complete networks as long as $\alpha$ is sufficiently small, significantly increasing the range of values of $\alpha$ for which a solution is known.  The solution to the continuous patrolling game remains open for many classes of networks for larger values of $\alpha$.

\section*{Acknowledgements} This material is based upon work supported by the National Science Foundation under Grant No. CMMI-1935826.

\section*{Appendix: proof of Proposition~\ref{prop:tightness}}\label{appendix:A}
\proof Let $w$ be an arbitrary patrol. Let $I_1=[0,6-\alpha]$, $I_2=[6-\alpha, \alpha]$, $I_3=[\alpha,6]$ and $I_4= [6, \infty)$. For $i=1,\ldots,4$, let $P(I_i)$ be the interception probability that $w(I_i)$ contributes to $P(w)$. Since the attack time is chosen uniformly at random in the interval $[0, 6-\alpha]$, all attacks are finished by time $6$ and $P(I_4)=0$. 

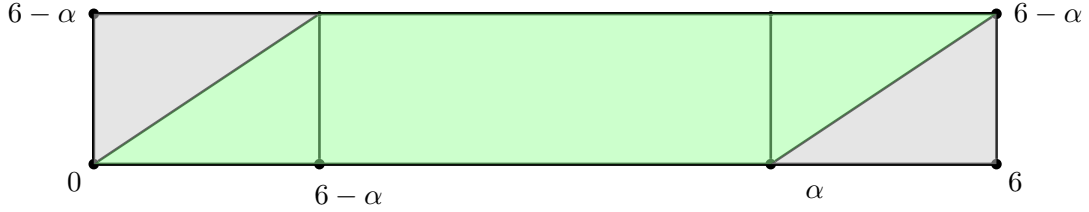
\begin{figure}[!ht]
    \centering
    \begin{tikzpicture}[line width=1pt]
    \tikzstyle{every node}=[draw,circle,fill=black,minimum size=1pt,
                            inner sep=0pt]
    \draw (0,0) node (0)[label=-150:${\ 0\ }$, minimum size=3pt] {}
        -- ++(0:3cm) node (1)[label=-30:${6-\alpha\ }$,  minimum size=3pt]  {}
        -- ++(0:6cm) node (2)[label=-30:${\ \ \ \ \alpha\ }$,  minimum size=3pt]  {}
        -- ++(0:3cm) node (3)[label=-30:${\ 6\ }$,  minimum size=3pt]  {}
        -- (0);
    \draw (0)
       -- ++(90:2cm) node (a)[label=-180:${\ 6-\alpha\ }$,  minimum size=3pt]  {}
       -- (0);
    \draw (a)
       -- ++(0:3cm) node (b)  {}
       -- ++(0:6cm) node (c)  {}
       -- ++(0:3cm) node (d)  [label=0:${\ 6-\alpha\ }$,  minimum size=3pt]{}
       -- (a);
    \draw (b) -- (1);
    \draw (c) -- (2);
    \draw (d) -- (3);
    \draw (0) -- (b);
    \draw (c) -- (b);
    \draw (d) -- (2);
    \fill[black!20, opacity=0.5] (0.center) -- (a.center) -- (b.center) --cycle;
    \fill[black!20, opacity=0.5] (d.center) -- (3.center) -- (2.center) --cycle;
    \fill[green!40, opacity=0.5] (0.center) -- (b.center) -- (d.center)-- (2.center) --cycle;
\end{tikzpicture}
    \caption{The maximum attacks that patrol $w$ can intercept in the time interval $[0,6]$. The starting time of the attack $[0, \alpha-6]$ is shown by the vertical lines. All attacks in the green area can be intercepted and all attacks in the grey area will not be intercepted.}
    \label{fig:2}
\end{figure}

We observe that $P(I_1) \leq (6-\alpha)/12$. Indeed, during $I_1$, the patrol can walk for length at most $6-\alpha$ without any point being revisited (see Figure~\ref{fig:2}) and that walk gives interception probability $(6-\alpha)/(2\mu)= (6-\alpha)/12$. Similarly, we have $P(I_2) \leq (2\alpha-6)/6$ and $P(I_3) \leq (6-\alpha)/12$ (see Figure~\ref{fig:2}). Then,
$$P(w)=P(I_1)+P(I_2)+P(I_3)\leq \frac{6-\alpha}{12}+\frac{2\alpha-6}{6}+\frac{6-\alpha}{12}=\frac{\alpha}{6}.$$
So, $P(w)=\alpha/6$ if and only if all $P(I_i)$ meet their bounds. In other words, in time $[0,6]$ the patrol $w$ must satisfy: (i) the patrol always walks with speed 1, and (ii) if any point $x$ is revisited, then $T_{j+1}(x)-T_j(x)\geq \alpha$ where $T_j(x)$ ($j=1,\ldots$) is the $j^{th}$ time $x$ is visited.

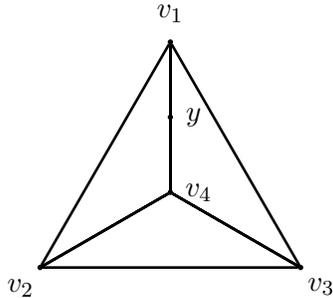
\begin{figure}[!ht]
    \centering
    \begin{tikzpicture}[line width=1pt]
    \tikzstyle{every node}=[draw,circle,fill=black,minimum size=1pt,
                            inner sep=0pt]
    \draw (0,0)[color = black] node (v4)[label=0:${\ v_4\ }$]{}
        -- ++(90:2cm) node (v1)[label=90:${\ v_1\ }$]{}
        -- (v4);
    \draw (v4)[color = black]
       -- ++(-30:2cm) node (v3)[label=-20:${\ v_3\ }$]{}
       -- (v4);
    \draw (v4)[color = black]
       -- ++(-150:2cm) node (v2)[label=-120:${\ v_2\ }$]{}
       -- (v4);
    \draw (v4)[color = black]
       -- ++(90:1cm) node (y)[label=0:${\ y\ }$]{}
       -- (v4);
    \draw [color = black](v2) -- (v3);
    \draw [color = black](v1) -- (v2);
    \draw [color = black](v1) -- (v3);
\end{tikzpicture}
\caption{Unit length arc network $K_4$.}
    \label{fig:Q_4}
\end{figure}

We claim that there is no patrol satisfying both (i) and (ii). First, assume at time 0, the patrol stays at a node. Since $\alpha >4$, $w([0,4])$ must be a path consisting of 4 distinct adjacent arcs. Without loss of generality, we consider 3 possible paths for $w([0,4])$: $w_1=(v_1, v_2, v_3, v_4, v_2)$, $w_2=(v_1, v_2, v_3, v_4, v_1)$, $w_3=(v_1, v_2, v_3, v_1, v_4)$ (see Figure \ref{fig:Q_4}). For $w_1$, to continue, the patrol can go to $v_3$ or $v_1$; however, both ways will immediately violate the condition (ii). For $w_2$, the patrols must continue by going from $v_1$ to $v_3$. Then, at $v_3$, there is no way to continue without violating the condition (ii). With the same analysis, $w_3$ cannot be completed such that the condition (ii) still holds.   

Second, we consider the case that the patrol starts at a regular point $y$. We assume $y \in (v_1, v_4)$ and the patrol first travels from $y$ to $v_1$ at time $t= d(y, v_1) < 1$. Since $t+3 <\alpha$, $w([0,t+3])$ cannot contain the same arc twice. It is enough to examine three possible cases for $w([0,t+3])$: 
\begin{itemize}
    \item Case 1: $w'=(y, v_1, v_2, v_3, v_1)$
    \item Case 2: $w''=(y, v_1, v_2, v_3, v_4)$
    \item Case 3: $w''=(y, v_1, v_2, v_4, v_3)$
\end{itemize}
Similar to the previous analysis, it is easy to see that in all cases condition (ii) cannot be satisfied.

For $\alpha=6$, it is easy to see that $V<1$ since there is no tour which cover all arcs in time $[0,6]$. So, for $4<\alpha$, the attack cannot be intercepted with probability $\alpha/\mu$ and the value of the game is $V< \alpha/\mu$. 
\endproof

\end{document}